\documentclass[runningheads]{llncs}
\usepackage[T1]{fontenc}
\usepackage{graphicx}
\usepackage{fluff-adt2024}
\newcommand{\pnm}{\mathrm{PNM}}
\newcommand{\snm}{\mathrm{SNM}}
\newcommand{\nm}[1]{\mathrm{NM}_{#1}}
\usepackage{natbib}
\begin{document}
\title{Toward Fair and Strategyproof Tournament Rules for Tournaments with Partially Transferable Utilities}
\titlerunning{Tournaments with Partially Transferable Utilities}
%
\author{David Pennock\inst{1}\orcidID{0000-0003-0522-4815} \and
Ariel Schvartzman\inst{2}\orcidID{0000-0003-4016-6235} \and
Eric Xue\inst{3}\orcidID{0009-0001-3977-8173}
}
\authorrunning{D. Pennock et al.}
%
\institute{DIMACS, Rutgers University, New Brunswick NJ 08901, USA \\
\email{dpennock@dimacs.rutgers.edu} \and
Google Research, Mountain View, CA 94043, USA \\
\email{aschvartzman@google.com} \and
Princeton University, Princeton NJ 08544, USA \\
\email{ex3782@princeton.edu}}
\maketitle              
\begin{abstract}
A tournament on $n$ agents is a complete oriented graph with the agents as vertices and edges that describe the win-loss outcomes of the $\binom{n}{2}$ matches played between each pair of agents.
The winner of a tournament is determined by a \emph{tournament rule} that maps tournaments to probability distributions over the agents.
We want these rules to be fair (choose a high-quality agent) and robust to strategic manipulation.
Prior work has shown that under minimally fair rules, manipulations between two agents can be prevented when utility is nontransferable but not when utility is completely transferable. We introduce a partially transferable utility model that interpolates between these two extremes using a selfishness parameter $\lambda$. 
Our model is that an agent may be willing to lose on purpose, sacrificing some of her own chance of winning, but only if the colluding pair's joint gain is more than $\lambda$ times the individual's sacrifice.

\hspace{10pt} 
We show that no fair tournament rule can prevent manipulations when $\lambda < 1$.
We computationally solve for fair and manipulation-resistant tournament rules for $\lambda = 1$ for up to 6 agents.
We conjecture and leave as a major open problem that such a tournament rule exists for all $n$.
We analyze the trade-offs between ``relative'' and ``absolute'' approximate strategyproofness for previously studied  rules and derive as a corollary that all of these rules require $\lambda \geq \Omega(n)$ to be robust to manipulation. 
We show that for stronger notions of fairness, non-manipulable tournament rules are closely related to tournament rules that witness decreasing gains from manipulation as the number of agents increases. 

\end{abstract}
\section{Introduction}

A tournament on $n$ agents is a complete oriented graph in which the agents are vertices and an edge from agent $i$ to agent $j$ means ``agent $i$ defeats agent $j$''.
These structures frequently arise in sports as the outcome of $\binom{n}{2}$ pairwise matches between $n$ agents or teams.
However, tournaments can arise 
whenever 
the performance of every two agents is comparable (e.g., agents are candidates in an election and edges are pairwise majority votes).

A tournament rule maps a tournament to a probability distribution over the agents.
These probabilities encode the likelihood that each agent is declared the tournament winner, or prescribe how to divide up a monetary reward~\cite{felsenthal1992}.
While a tournament rule should be fair in that it chooses some qualified agent who beats many other agents, it also should not reward manipulations: for example, losing a match on purpose should not improve an agent's or their co-conspirator's chances of winning the tournament.
If the rule is manipulable, then agents may act in ways that undermine the primary goal of choosing a highly qualified winner. 
In fact, instances of these actions are not unheard of in sports.
At the London 2012 Olympic Games, four women's doubles teams were disqualified for attempting to throw their final matches in the round-robin group stage in order to earn a more favorable seed in the knockout stage of the tournament.

Unfortunately, prior work has shown that fairness and non-manipulability are largely incompatible.
A prevailing notion of fairness studied by prior work~\cite{altman2010nonmanipulable, altman2009nonmanipulable, ding2021approximately, schneider2016condorcetconsistent, schvartzman2019approximately} is Condorcet consistency.
A tournament rule is Condorcet consistent if, whenever one agent beats all other agents, the undefeated agent wins the tournament with certainty.

Altman, Procaccia, and Tennenholtz~\cite{altman2009nonmanipulable} showed that any deterministic rule that satisfies this notion is susceptible to pairwise manipulations: 
for any Condorcet consistent rule, there exist tournaments in which two agents can influence the choice of winner by colluding to reverse the outcome of their match.

Altman and Kleinberg \cite{altman2010nonmanipulable} extended this work to randomized rules that map tournaments to probability distributions over agents.
They showed that there exist Condorcet consistent and pairwise non-manipulable rules when two agents collude only if one of them can strictly improve her probability of winning at no cost to the other.
Rules that are pairwise non-manipulable under this assumption are said to be 2-Pareto non-manipulable (2-$\pnm$).
However, no Condorcet consistent rule exists when utility is completely transferable---that is, when two agents only care about the probability that at least one of them wins the tournament.
Instead, the authors 
demonstrated rules that are approximately Condorcet consistent and pairwise non-manipulable in this setting, which the authors term 2-strongly non-manipulable (2-$\snm$).
Another line of work~\cite{ ding2021approximately, schneider2016condorcetconsistent, schvartzman2019approximately} sought rules that were fair and approximately 2-$\snm$.

Motivated by the fact that collusion and the deliberate throwing of matches in sports occur less frequently than the negative results of prior work imply, we extend prior work to the setting in which utility is partially transferable.
These settings are natural.
For example, consider a setting in which a tournament rule is used to prescribe a division of monetary reward among the participants.
Because the reward is divisible, if two agents could improve their share of the reward by fixing the outcome of their match, then they may choose to do so and redistribute their winnings later so that the collusion is mutually beneficial. 
But collusions are rarely so frictionless in reality.
There could be uncertainty as to whether the agent that benefits from fixing the outcome of the match will follow through with the redistribution.
There could be penalties for agents found to have thrown their matches.
Or there could be factors beyond the outcome of the tournament that matter, such as an agent's reputation, that a loss would negatively affect.
With these frictions, agents would not be completely altruistic to their partner (fully transferable utility) nor completely selfish (non-transferable utility). 
Instead, agents would care more about winning themselves but may be willing to sacrifice their own probability if it achieves a significant proportional gain for their partner.

We model each agent's values for her own probability of winning \emph{and} for her collusion partner's probability of winning as being in some ratio and extend prior notions of non-manipulability by introducing a term that accounts for the range of selfishness of agents.
More specifically, we say a rule is 2-$\nm{\lambda}$ if no agent can collude with another to improve her probability of winning by at least a $\lambda + 1$ factor of the decrease in probability witnessed by her colluding partner.
Stated another way, under a 2-$\nm{\lambda}$ rule, no pairwise collusions would occur if we assume that each agent would not sacrifice her own chances of winning unless her partner gains at least $\lambda + 1$ times the amount that she loses.

We show that this model connects the notions of Pareto and strong non-manipulability by varying $\lambda$.
Moreover, we conjecture that there exists a tournament rule that is monotone, Condorcet consistent, and 2-$\nm{1}$, implying that it is possible to prevent deliberate loss and collusion, as long as each agent weighs her own probability of winning twice as much as her opponents'.
However, we show that none of the rules proposed in five previous papers~\cite{altman2010nonmanipulable, pagerank-tournament-rule, ding2021approximately, schneider2016condorcetconsistent, schvartzman2019approximately} satisfy this combination of conditions by demonstrating how these rules trade-off between $\lambda$, our notion of relative approximate strategyproofness, and 
the established notion of absolute approximate strategyproofness~\cite{schneider2016condorcetconsistent}.

In a separate direction, we introduce another notion of fairness, termed dominant sub-tournament consistency (DSTC), and show that several natural rules satisfy this condition.
Intuitively, a rule is DSTC if the addition of an agent that loses to the 
original 
agents does not affect their probabilities.
A closely related notion is top cycle consistency (TCC), which requires the winner to come from the top cycle with certainty.
We show that within these notions of fairness, the problem of finding a rule that is 2-$\nm{\lambda}$ reduces to the problem of finding a rule that witnesses gains from manipulation that vanishes as the number of agents increases. 

\subsection{Related Work}

For a broad discussion of recent developments on tournaments in computational social choice, see Suksompong's excellent survey~\cite{WarutSurvey21}.
We discuss work closely related to ours. 

Altman and Kleinberg~\cite{altman2010nonmanipulable} and %
Altman, Procaccia, and Tennenholtz~\cite{altman2009nonmanipulable} %
were the first to consider the question of strategic manipulations of tournaments by agents. Their main conclusion is that Condorcet consistency and strong non-manipulability are directly at odds: no tournament rule, even randomized ones, can satisfy both properties. Later, Schneider, Schvartzman, and Weinberg~\cite{schneider2016condorcetconsistent} considered a relaxation of the problem: they sought tournament rules that are Condorcet consistent and are minimally manipulable. Their main result is that the Randomized Single Bracket Elimination (RSEB) rule is 2-$\snm$-1/3, meaning that the most probability that any pair can gain is 1/3, and this is optimal among all Condorcet-consistent rules. This result was later strengthened to show that the Randomized King of the Hill (RKotH) rule is also 2-$\snm$-1/3 and cover consistent, a notion strictly stronger than Condorcet consistent~\cite{schvartzman2019approximately}. 
Recent discoveries include a rule that is 3-$\snm$-31/60, meaning that the most probability that any coalition of three agents can gain is 31/60, the first explicit rule that is 3-$\snm$-$\alpha$ for $\alpha < 1$~\cite{DinevW22},
and a different rule that is 3-$\snm$-1/2 ~\cite{Miksanik23}.
Parallel lines of work have considered variations on this problem,
including probabilistic tournaments~\cite{ding2021approximately} and tournaments with prize vectors for multiple places rather than only one prize for the winner~\cite{DaleFRSW22}. 


\section{Preliminaries}

\begin{definition}[Tournament]
A tournament $T = (A, \succ_T)$ is a pair where $A$ is a finite set of agents and $\succ_T$ is a complete asymmetric binary relation over $A$ that describes the outcomes of the $\binom{\abs{A}}{2}$ matches played between each pair of distinct agents.
For agents $i \not= j \in A$, we write $i \succ_T j$ if $i$ dominates $j$ in $T$.
Let $\mathcal{T}_n$ denote the set of tournaments where $[n]$ is the set of agents.
\end{definition}

\begin{definition}[Tournament rule]
A tournament rule on $n$ agents $r^{(n)}: \mathcal{T}_n \to \Delta^n$ maps a tournament $T \in \mathcal{T}_n$ to a probability distribution over the agents.
A tournament rule $r$ is a family of tournament rules on $n$ agents $\{r^{(n)}\}_{n=1}^\infty$.
For all $n \in \NN$ and $T \in \mathcal{T}_n$, we write $r(T) := r^{(n)}(T)$, and for $i \in [n]$, we write $r_i(T)$ to denote the probability that $i$ wins $T$ under $r$.
\end{definition}

\subsection{Fairness Properties}

A desirable tournament rule should choose the most qualified agent as the winner of a tournament.
In line with this reasoning, we want a tournament rule to choose an undefeated agent with probability 1 since this agent is clearly better than the rest of her opponents.

\begin{definition}[Condorcet consistency]
A tournament rule on $n$ agents $r^{(n)}$ is Condorcet consistent (CC) if for all $T \in \mathcal{T}_n$, $r_i^{(n)}(T) = 1$ whenever there exists $i \in [n]$ such that $i \succ_T j$ for all $j \in [n] \setminus \{i\}$.
A tournament rule $r$ is CC if $r^{(n)}$ is CC for all $n$.
\end{definition}

Note that Condorcet consistency is quite a minimal notion of fairness since it is binding only when there is an agent that is clearly superior than the others.
Unfortunately, it is often the case that no such agent exists.
The following notions of fairness seek to restrict the subset of agents that should be named the winner in such cases by eliminating those who are in some sense clearly worse than her opponents.

\begin{definition}[Top cycle consistency]
A subset of agents $S$ is the top cycle in tournament $T$ if it is the minimal subset of agents such that $i \succ_T j$ for all $i \in S, j \in [n] \setminus S$.
The top cycle of a tournament always exists and is unique.
Let $TC(T)$ denote the top cycle of $T$.
A tournament rule on $n$ agents $r^{(n)}$ is top cycle consistent (TCC) if for all $T \in \mathcal{T}_n$, $r_i^{(n)}(T) = 0$ for all $i \in [n] \setminus TC(T)$.
A tournament rule $r$ is TCC if $r^{(n)}$ is TCC for all $n$.
\end{definition}

Top cycle consistency extends Condorcet consistency quite naturally: Condorcet consistency requires that an undefeated agent be declared the winner, while top cycle consistency requires this winner to come from the smallest undefeated subset.
Moreover, since the agents in the top cycle are undefeated by those outside of the top cycle, they are in some sense better.
On the other hand, no agent in the top cycle is clearly superior than the others since every agent in the top cycle is defeated by another in the top cycle.

\begin{definition}[Cover consistency]
For $i \not= j$, we say $i$ covers $j$ if $i \succ_T j$ and $j \succ_T k \implies i \succ_T k$ for all $k \in [n] \setminus \{i,j\}$. 
Moreover, we say $j$ is covered if there exists $i \in [n]$ such that $i$ covers $j$.
A tournament rule on $n$ agents $r^{(n)}$ is cover consistent if for all $T \in \mathcal{T}_n$, $r_j^{(n)}(T) = 0$ whenever $j$ is covered.
A tournament rule $r$ is cover consistent if $r^{(n)}$ is cover consistent for all $n$.
\end{definition}

Cover consistency refines top cycle consistency by further restricting the set of potential winners.
If $i$ covers $j$, then not only did $i$ defeat $j$, but $i$ also defeated everyone that $j$ defeated.
Thus, covered agents are worse than the agents that cover them in some sense.

\begin{definition}[Dominant sub-tournament consistency]
For a subset of agents $S \subseteq [n]$ and a tournament $T \in \mathcal{T}_n$, let $T|_S = (S, \{(i,j) \in S \times S : i \succ_T j\})$ denote the subgraph induced by $S$.
$T|_S$ is a dominant sub-tournament in tournament $T$ if $i \succ_T j$ for all $i \in S, j \in [n] \setminus S$.
A tournament rule $r$ is dominant sub-tournament consistent (DSTC) if $r_i(T|_S) = r_i(T)$ for all $i \in S$. 
\end{definition}

Dominant sub-tournament consistency strengthens top cycle consistency in a different direction than cover consistency.
Rather than narrow down the set of potential winners, dominant sub-tournament consistency requires that the probability of choosing a certain member of the top cycle as the winner is the same as the probability of choosing her if the agents outside the top cycle were removed.
To the best of our knowledge, dominant sub-tournament consistency has not been considered before in the tournament literature.

The following result formalizes the hierarchy of fairness conditions.

\begin{proposition}[Fairness hierarchy]
Any tournament rule that satisfies either cover consistency or DSTC satisfies TCC.
Moreover, any TCC rule is CC.
\end{proposition}

\begin{proof}
Let $T \in \mathcal{T}_n$.
Suppose $r$ is a cover consistent tournament rule, and consider any $j \not\in TC(T)$.
Observe that any $i \in TC(T)$ covers $j$ since by definition of the top cycle, we have that $i \succ_T j$, and for any $k \in [n]$ such that $j \succ_T k$, we have that $k \not\in TC(T)$, so $i \succ_T k$.
Since $r$ is cover consistent and $j$ is covered, we have that $r_j(T) = 0$.
Thus, $r$ is TCC.

Now, suppose $r$ is a DSTC tournament rule.
By definition, $TC(T)$ is a dominant sub-tournament of $T$.
Thus, $\sum_{i \in TC(T)} r_i(T) = \sum_{i \in TC(T)} r_i(T|_{TC(T)}) = 1$.
It follows  that $r_i(T) = 0$ for all $i \in [n] \setminus TC(T)$, so $r$ is TCC.

Now, suppose $r$ satisfies TCC, and note that whenever some agent $i$ is undefeated in $T$, $i$ is the only member of $TC(T)$: $i$ dominates every $j \not= i$, and no proper subset of $\{i\}$ satisfies this property.
Thus, $r_i(T) = 1$ and $r$ is CC.
\qed
\end{proof}

\subsection{Non-manipulability Properties}

In addition to satisfying some notion of fairness, tournament rules should be robust to manipulation.
In this work, we consider manipulations where a single agent purposefully loses her match against one of her opponents and manipulations where two agents collude to reverse the outcome of their match.

\begin{definition}[$S$-adjacent]
$T, T' \in \mathcal{T}_n$ are $S$-adjacent where $S \subseteq [n]$ if $i \succ_T j \iff i \succ_{T'} j$ for $i \not= j \in [n] \setminus S$.
In other words, $T$ and $T'$ are $S$-adjacent if they coincide on every match except possibly those between agents in $S$.
\end{definition}

When utilities are nontransferable, two agents are willing to collude only if one of them can strictly improve her probability of winning at no cost to the other.
Formally, distinct agents $i, j \in [n]$ collude from tournament $T$ to tournament $T'$ only if $\max\{r_i(T') - r_i(T), r_j(T') - r_j(T)\} > 0$ and $\min\{r_i(T') - r_i(T), r_j(T') - r_j(T)\} \geq 0$.
Thus, to incentivize agents against such manipulations, a tournament rule must satisfy the following notion of non-manipulability.

\begin{definition}[2-Pareto non-manipulability]
A tournament rule $r$ is 2-Pareto non-manipulable (2-PNM) if for all $i \not= j \in [n]$ and $\{i, j\}$-adjacent tournaments $T \not= T' \in \mathcal{T}_n$, either (1) $\min\{r_i(T') - r_i(T), r_j(T') - r_j(T)\} < 0$ or (2) $\max\{r_i(T') - r_i(T), r_j(T') - r_j(T)\} \leq 0$.
\end{definition}

Altman and Kleinberg~\cite{altman2010nonmanipulable} give a rule that is monotone, TCC, and 2-PNM.
The barrier to pairwise manipulation is much lower when utilities are completely transferable since two agents only care about the probability that at least one of them wins the tournament.
In other words, $i$ and $j$ collude from $T$ to $T'$ only if $r_i(T') + r_j(T') > r_i(T) + r_j(T)$.
Under this utility model, an agent may be willing to sacrifice and shift a significant portion of her probability to her partner in crime.
Thus, tournament rules must satisfy a stronger notion of non-manipulability in this setting.

\begin{definition}[2-strong non-manipulability]
A tournament rule $r$ is 2-strongly non-manipulable (2-SNM) if $r_i(T') + r_j(T') \leq r_i(T) + r_j(T)$ for all $i \not= j \in [n]$ and $\{i, j\}$-adjacent tournaments $T \not= T' \in \mathcal{T}_n$.
\end{definition}

Prior work has shown that no Condorcet consistent tournament rule is 2-SNM.
However, despite this strong impossibility result, instances of collusion are relatively infrequent in the real world, suggesting that settings in which utilities are completely transferable are uncommon.
On the other hand, instances of collusion are not unheard of, suggesting that utility is neither always nontransferable.

In this paper, we consider a third utility model in which utilities are partially transferable: distinct agents $i$ and $j$ collude from tournament $T$ to tournament $T'$ only if $r_i(T') + r_j(T') > r_i(T) + r_j(T) + \lambda \max\{r_i(T) - r_i(T'), r_j(T) - r_j(T')\}$.
In this model, two agents always collude if both of them improve their chances of winning and never collude if both of their chances decrease.
The interesting case is when one agent improves her chances at the expense of the other.
One interpretation of this necessary condition is that agents would rather win the tournament themselves but are still willing to collude if the gain in probability is significantly larger than each agent's loss.
Here, $\lambda$ is a parameter that measures how transferable utility is.
Note that when $\lambda$ is low, utilities are more transferable.
We will later see how $\lambda$ can be interpreted as agents' level of selfishness. 
We now define a notion of non-manipulability for this model.


\begin{definition}[2-non-manipulability for $\lambda$]
A tournament rule $r$ is 2-non-manipulable for $\lambda \geq 0$ (2-$\nm{\lambda}$) if $r_i(T') + r_j(T') \leq r_i(T) + r_j(T) + \lambda \max\{r_i(T) - r_i(T'), r_j(T) - r_j(T')\}$ for all $i \not= j \in [n]$ and $\{i, j\}$-adjacent tournaments $T \not= T' \in \mathcal{T}_n$.
We say $r$ is 2-$\nm{\infty}$ if $r_i(T') + r_j(T') \leq r_i(T) + r_j(T) + \lim_{\lambda \to \infty} \lambda \max\{r_i(T) - r_i(T'), r_j(T) - r_j(T')\}$ (where the limit is taken in the extended reals).
\end{definition}

Observe that when $\lambda = 0$, our notion of non-manipulability coincides with strong non-manipulability.
Moreover, we show that our notion coincides with Pareto non-manipulability when $\lambda = +\infty$.
We remark that we do not interpret 2-$\nm{\lambda}$ as an approximation to 2-$\snm$.
Unlike approximation algorithms, a tournament designer who finds herself faced with e.g., agents who value their opponents chances of winning as much as their own (completely transferable utility) may not find it in her best interest to use a 2-$\nm{\lambda}$ tournament rule for some $\lambda > 1$.
Rather, $\lambda$ is meant to model the behavior of the agents.

\begin{proposition}\label{prop:pnm-iff-nm-infty}
A tournament rule is 2-$\pnm$ if and only if it is 2-$\nm{\infty}$.
\end{proposition}

By Proposition \ref{prop:pnm-iff-nm-infty}, our notion of non-manipulability generalizes strong and Pareto non-manipulability while connecting the two.
As in previous work~\cite{schneider2016condorcetconsistent, schvartzman2019approximately}, we are interested in approximately non-manipulable tournament rules; that is, rules under which no two agents can collude to gain in joint probability more than $\alpha$ more than each agent's loss (weighted by $\lambda$).
We will see later that there is a range of $\lambda$ for which fair and non-manipulable tournament rules do not exist.
For $\lambda$ in this range, it may be better to design approximately non-manipulable tournament rules tailored to $\lambda$ than use a 2-$\nm{\lambda'}$ tournament rule for some $\lambda' > \lambda$.

\begin{definition}[2-non-manipulability up to $\alpha$ for $\lambda$]
A tournament rule $r$ is 2-non-manipulable up to $\alpha$ for $\lambda \geq 0$ (2-$\nm{\lambda}$-$\alpha$) if $r_i(T') + r_j(T') \leq r_i(T) + r_j(T) + \lambda \max\{r_i(T) - r_i(T'), r_j(T) - r_j(T')\} + \alpha$ for all $i \not= j \in [n]$ and $\{i, j\}$-adjacent tournaments $T \not= T' \in \mathcal{T}_n$.
\end{definition}


In addition to being robust against pairwise manipulations, a tournament rule should be robust to the intentional throwing of matches.

\begin{definition}[Monotonicity]
A tournament rule is monotone if $r_i(T) \geq r_i(T')$ for all $i \not= j \in [n]$ and $\{i,j\}$-adjacent tournaments $T \not= T' \in \mathcal{T}_n$ such that $i \succ_T j$.
\end{definition}

Intuitively, monotonicity says that no agent should be able to improve her chances of winning by deliberately losing one of her matches.
Thus, agents have an incentive to win each of their matches under monotone rules.
Violations of this property should be seen as quite severe.

\begin{proposition}\label{prop:mono-iff-one-sided-nm}
Let $r$ be a 2-$\nm{\lambda}$ tournament rule for some $\lambda > 0$, then the following two statements are equivalent.
\begin{enumerate}
    \setlength{\itemsep}{0pt}
    \item $r$ is monotone
    \item For all $i \not= j \in [n]$ and $\{i, j\}$-adjacent tournaments $T \not= T' \in \mathcal{T}_n$ such that $i \prec_T j$, $r_i(T') - r_i(T) \leq (\lambda + 1) (r_j(T) - r_j(T'))$
\end{enumerate}
\end{proposition}

Proposition \ref{prop:mono-iff-one-sided-nm} offers a natural interpretation of the parameter $\lambda$ and the 2-$\nm{\lambda}$ property for monotone tournament rules: $\lambda$ is how much each agent weighs her own probability of winning over others' probabilities of winning and a tournament rule is 2-$\nm{\lambda}$ if switching the outcome of a match does not increase the probability of winning for the new winner by more than a $\lambda + 1$ factor over the loss of the new loser.
Note that Proposition \ref{prop:mono-iff-one-sided-nm} does not hold for $\lambda = 0$.
Indeed, monotonicity and 2-$\snm$ are independent properties: neither implies the other.

\subsection{Tournament Rules}

In this section, we define several tournament rules.
See Table~\ref{tab:summary-prior-results} for a summary of what was known about them prior to this work (to the best of our knowledge).

\begin{table}
    \caption{Summary of relevant prior results. Unless stated otherwise, results come from the paper that proposed the tournament rule.}
    \centering
    \begin{tabular}{ c || c | c | c | c}
        Rule & Monotone? & Fairness & 2-$\pnm$? & 2-$\snm$-$\alpha$  \\
        \hline
        ICR~\cite{altman2010nonmanipulable} & Yes & TCC & Yes & $\alpha \geq \frac{1}{2} - \frac{1}{n(n-1)}$~\cite{schneider2016condorcetconsistent} \\
        RVC~\cite{altman2010nonmanipulable} &  Yes & TCC & Yes & $\alpha \geq \frac{1}{2} - \frac{n-3}{n(n-1)}$~\cite{schneider2016condorcetconsistent} \\
        TCR~\cite{altman2010nonmanipulable} & Yes & TCC & Yes & $\alpha \geq 1-2/n$~\cite{schneider2016condorcetconsistent} \\
        RSEB~\cite{schneider2016condorcetconsistent} & Yes & CC & ? & $\alpha = 1/3$ \\
        RKotH~\cite{schvartzman2019approximately} & Yes & cover & ? & $\alpha = 1/3$ \\
        RDM~\cite{ding2021approximately} & ? & CC & ? & $\alpha = 1/3$ \\
        PR~\cite{pagerank-tournament-rule} & ? & ? & ? & ? \\
        PRSL~\cite{pagerank-tournament-rule} & ? & ? & ? & ?
    \end{tabular}
    \label{tab:summary-prior-results}
\end{table}

\begin{enumerate}
\item The \textbf{Iterative Condorcet Rule} (ICR) chooses the undefeated agent if one exists.
Otherwise, eliminate an agent uniformly at random and repeat.
\item The \textbf{Randomized Voting Caterpillar} rule (RVC) begins by choosing a permutation of the agents uniformly at random.
In the first iteration, RVC eliminates the loser between the first and second agents in the permutation.
In each subsequent iteration until only one agent remains, RVC eliminates the loser between the previous winner and the next agent in the permuation.
\item The \textbf{Top Cycle Rule} (TCR) chooses an agent uniformly at random from the top cycle and declares her the winner.
\item A single elimination bracket is a complete binary tree whose leaves are labeled by a permutation of the agents.
Each node is labeled by the winner of the match between its two children.
The winner of the bracket is the agent labeling the root node.
The \textbf{Randomized Single Elimination Bracket} rule (RSEB) introduces $2^{\ceil{\log n}} - n$ dummy agents who lose to the existing agents, chooses a bracket uniformly at random, and declares the winner of this bracket the winner of the tournament.
\item The \textbf{Randomized King of the Hill} rule (RKotH) chooses the undefeated agent if one exists.
Otherwise, choose an agent uniformly at random, eliminate her and the agents she dominates, and repeat.
\item The \textbf{Randomized Death Match} rule (RDM) chooses a pair of agents uniformly at random, eliminates the loser, and repeats.
\item The \textbf{PageRank} (PR) (\textbf{with Self-Loops} (PRSL)) rule chooses agent $i$ as the winner with probability
\[
    r_i(T) = \begin{cases}\mathrm{pr}_i\Paren{T|_{TC(T)}} & i \in TC(T) \\ 0 & i \not\in TC(T)\end{cases}
\]
where for a strongly connected (sub)tournament $S$, $\mathrm{pr}(S)$ is the unique solution to the following linear system of equations: 
\begin{align*}
    \forall\,i,\:\mathrm{pr}_i(S) &= \sum_{j : i \succ_{S} j} \frac{1}{\abs{\{k : j \prec_S k\}} + \1(PRSL)} \mathrm{pr}_j(S) \\
    \textstyle \sum_i \mathrm{pr}_i(S) &= 1
\end{align*}
Note that both PR and PRSL are well-defined since the top cycle is strongly connected, so the stationary distribution is indeed unique.
PageRank's recursive definition is natural for tournaments: an agent has high PageRank if she beats many other agents with high PageRank.
\end{enumerate}

We show that many previously studied tournament rules actually satisfy stronger notions of fairness than previously  demonstrated.
We particularly highlight that many of them satisfy our proposed notion of DSTC.

\begin{theorem}
RSEB satisfies TCC but neither DSTC nor cover consistency.
ICR, RVC, TCR, RDM, PR and PRSL satisfy DSTC but not cover consistency.
RKotH satisfies both DSTC and cover consistency.
ICR, RVC, TCR, RSEB, RKotH, and RDM are monotone.
\end{theorem}


\begin{proof}
Schvartzman et al.~\cite{schvartzman2019approximately} showed that RKotH satisfies cover consistency.
The authors of~\cite{altman2010nonmanipulable, schneider2016condorcetconsistent, schvartzman2019approximately} proved the monotonicity of ICR, RVC, TCR, RSEB, and RKotH.
RDM is monotone since for any deterministic sequence of matches, an agent gets at least as far as she did in the original tournament if she wins an additional match.

RSEB satisfies TCC because in order for an agent outside of the top cycle to win, she must eventually defeat an agent in the top cycle.
To see how RSEB violates DSTC and cover consistency, consider the 8-agent tournament $T$ in which 
\begin{align*}
    1 \succ_T 2, 2 \succ_T 3, 3 \succ_T 4, 4 \succ_T 1, 1 \succ_T 3, 2 \succ_T 4 
\end{align*}
and $i \succ_T j$ for all $i \in \{1, 2, 3, 4\}, j \in \{5, 6, 7, 8\}$.
The relations between $\{5,6,7,8\}$ can be arbitrary.
Note that 2 covers 3 in $T$ yet 3 can win e.g., the bracket whose leaves are labeled by the permutation $(1, 2, 4, 5, 3, 6, 7, 8)$.
Moreover, the sub-tournament induced by the first four agents is a dominant sub-tournament.
Observe that 4 never wins a bracket in $T|_{[4]}$, but in $T$, 4 can win e.g., the bracket whose leaves are labeled by the permutation $(1, 2, 3, 5, 4, 6, 7, 8)$.

\begin{figure}
    \centering
    \includegraphics[width=\linewidth]{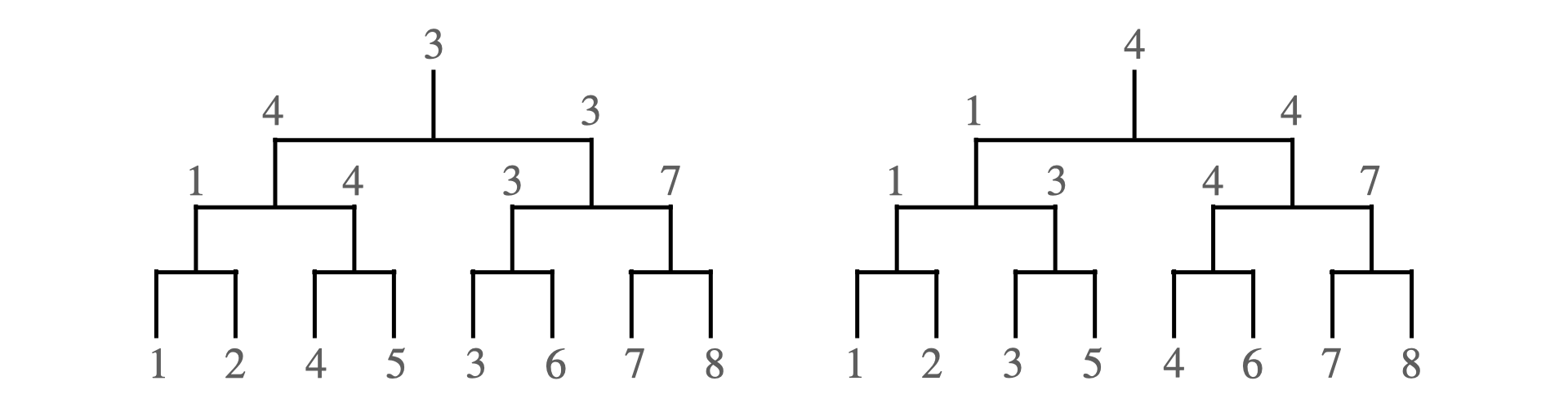}
    \caption{3 is covered yet wins in the left bracket. 4 never wins a bracket in $T|_{[4]}$ yet wins in the right bracket.
    }
\end{figure}

TCR satisfies DSTC because the addition of an agent that loses to all existing agents (a Condorcet loser) does not change the top cycle.
ICR, RVC, and RKotH satisfy DSTC because inserting a Condorcet loser into a permutation does not change the winner, so each agent wins the same proportion of permutations they did before.
Similarly, RDM satisfies DSTC because inserting a match involving a Condorcet loser into a sequence of matches does not change the winner.
PR and PRSL are DSTC by definition.

To see why ICR, RVC, TCR, RDM, PR, and PRSL fail to satisfy cover consistency, it suffices to consider $T|_{[4]}$.
3 is in the top cycle, so she can win with positive probability under TCR, PR, and PRSL.
3 wins ICR and RVC if the chosen permutation is $(2, 1, 4, 3)$.
3 wins RDM if $(1,2)$ is the first pair chosen, $(1,4)$ is the second, and $(3,4)$ is the last.
\qed
\end{proof}

\section{Lower Bounds}

Schneider et al.~\cite{schneider2016condorcetconsistent} showed that no Condorcet consistent tournament rule is 2-$\nm{0}$-$\alpha$ for $\alpha < 1/3$.
The same lower bound construction yields Theorem \ref{thm:lower-bound-all-rules}.

\begin{theorem}\label{thm:lower-bound-all-rules}
No CC tournament rule is 2-$\nm{\lambda}$-$\alpha$ for $\lambda < 1 - 3\alpha$.
\end{theorem}

\begin{proof}
We prove the theorem for monotone rules, but with some additional casework, one can extend the result to non-monotone rules.

Suppose tournament rule $r$ is CC and 2-$\nm{\lambda}$-$\alpha$, and consider any tournament $T$ on $[n]$ in which $1$ dominates $2$, $2$ dominates $3$, and $3$ in turn dominates $1$.
Note that any two agents among $\{1,2,3\}$ can collude so that one of them becomes undefeated.
Since $r$ is monotone, CC, and 2-$\nm{\lambda}$-$\alpha$,
\begin{align*}
    1 - r_2(T) &\leq (\lambda + 1) r_1(T) + \alpha \\
    1 - r_3(T) &\leq (\lambda + 1) r_2(T) + \alpha \\
    1 - r_1(T) &\leq (\lambda + 1) r_3(T) + \alpha 
\end{align*}
Adding these three inequalities together and isolating $\lambda$ yields
\[
    \textstyle \frac{3(1 - \alpha)}{r_1(T) + r_2(T) + r_3(T)} - 2 \leq \lambda
\]
Since $r_1(T) + r_2(T) + r_3(T) \leq 1$, this inequality implies that $\lambda \geq 1 - 3\alpha$.
\qed
\end{proof}

\begin{corollary}
No Condorcet consistent tournament rule is 2-$\nm{\lambda}$ for $\lambda < 1$.
\end{corollary}

We believe this lower bound is tight.
That is, we believe that there exists a monotone and Condorcet consistent tournament rule that is 2-$\nm{1}$.
Thus, pairwise collusion can be prevented without sacrificing fairness as long as agents prefer not to collude if their sacrifice in probability is greater than the joint gain. Figure~\ref{fig:non-isomorphic-4} shows such an $r$ for 4 agents.
Expressing and computationally solving the problem as a feasibility linear program show that such rules exist for tournaments of up to 6 agents. 
Unfortunately, as the number of tournaments on $n$ agents grows exponentially with $n$, it became computationally difficult to check whether such rules exist for tournaments of larger size.

\begin{figure}
    \centering
    \includegraphics[width=\linewidth]{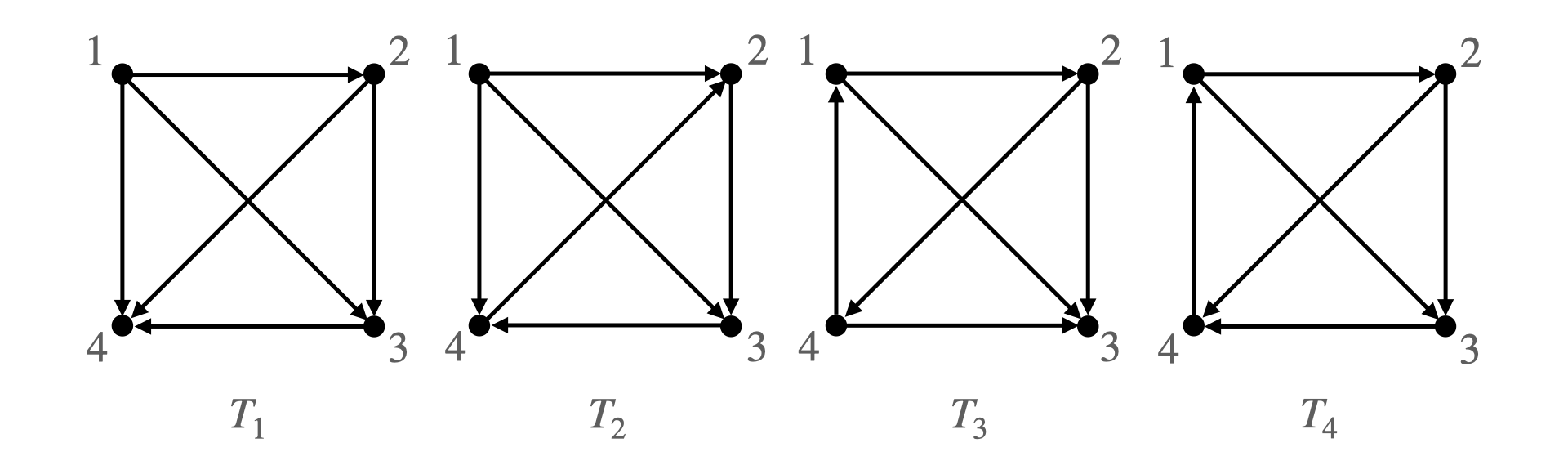}
    \caption{All non-isomorphic tournaments on 4 agents. 
        The following conditions are necessary and sufficient for a tournament rule on 4 agents $r$ to be Condorcet-consistent and 2-$\nm{1}$.
        In $T_1$ and $T_2$, $r$ chooses 1 as the winner with probability 1. 
        In $T_3$, $r$ chooses the winner uniformly at random among 1, 2, and 4.
        In $T_4$, $r$ chooses the winner according to a distribution that is a convex combination of $\left(\frac{4}{9}, \frac{2}{9}, 0, \frac{3}{9} \right)$, $\left(\frac{5}{9}, \frac{1}{9}, 0, \frac{3}{9} \right)$, $\left(\frac{13}{33}, \frac{8}{33}, \frac{2}{33}, \frac{10}{33} \right)$, $\left(\frac{5}{12}, \frac{13}{48}, \frac{1}{48}, \frac{7}{24} \right)$, $\left(\frac{11}{21}, \frac{4}{21}, \frac{1}{21}, \frac{5}{21} \right)$, and $\left(\frac{17}{39}, \frac{7}{39}, \frac{4}{39}, \frac{11}{39} \right)$. 
    }
    \label{fig:non-isomorphic-4}
\end{figure}

\begin{conjecture}
There exists a tournament rule that is monotone, Condorcet consistent, and 2-$\nm{1}$.
\end{conjecture}

We now consider several tournament rules and examine their trade-offs between $\alpha$ and $\lambda$.
Table \ref{tab:summary} provides a summary of our findings.
Interestingly, the superman-kryptonite tournament identified by Schneider et al.~\cite{schneider2016condorcetconsistent} (and its variants) is responsible for all our lower bounds, suggesting that it is especially problematic.
We note that Iglesias et al.~\cite{votingtree} identified a general class of tournaments termed perfect manipulator tournaments that contains the superman-kryptonite tournament while studying a different problem.

\begin{definition}[Superman-kryptonite tournament]
The superman kryptonite tournament on $[n]$ has $i \succ_T j$ whenever $i < j$, except $n \succ_T 1$.
In particular, superman $1$ dominates all agents but kryptonite $n$, and $n$ is dominated by all agents except $1$.
\end{definition}

\begin{figure}
    \centering
    \includegraphics[width=\linewidth]{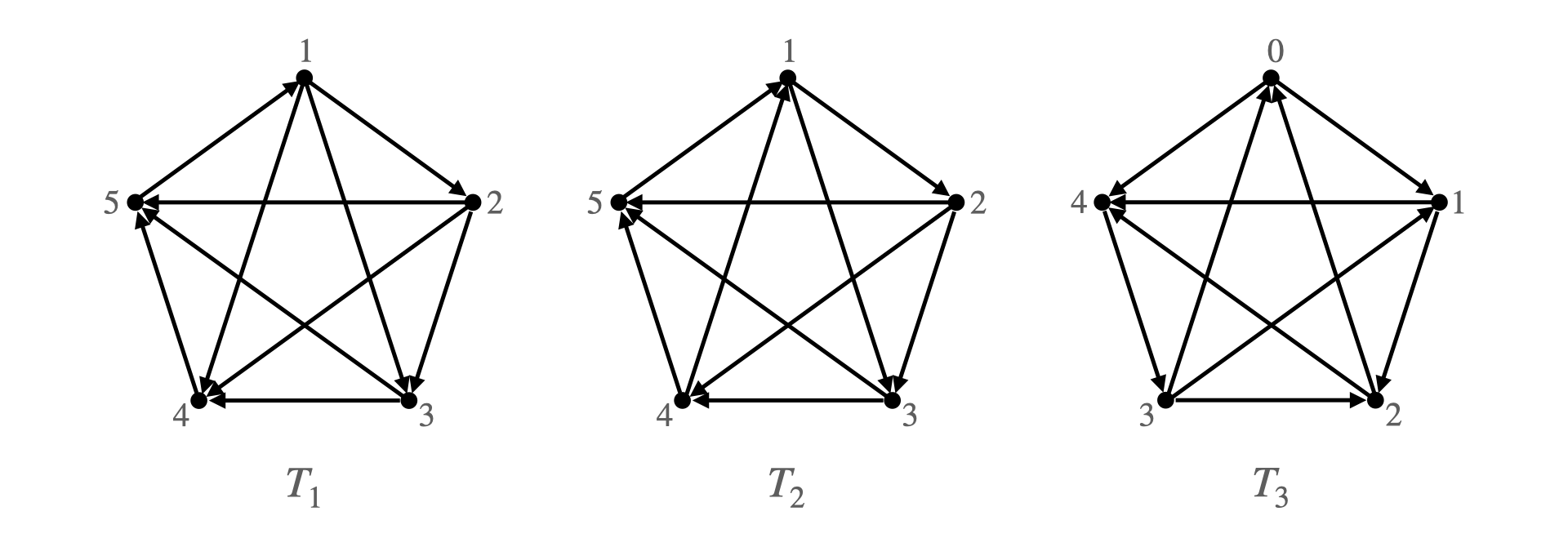}
    \caption{The tournaments that lead to our lower bounds. 
        $T_1$ is the superman-kryptonite tournament (here, on 5 players) that leads to most of our lower bounds.
        $T_2$ is the tournament that leads to the $\alpha \geq 1/10$ lower bound for RKotH.
        $T_3$ is the tournament (here, on 5 players) that leads to our lower bound for PRSL.
    }
    \label{fig:lower-bound-tournaments}
\end{figure}

\begin{theorem}\label{theorem:RSEB-RKotH-no-lambda}
Let $\lambda \geq 0$.
If RSEB satisfies 2-$\nm{\lambda}$-$\alpha$, then $\alpha \geq \Omega(1/n)$.
If RKotH or PR satisfy this property, then $\alpha \geq \Omega(1)$.
That is, RSEB, RKotH, and PR are always pairwise manipulable regardless of $\lambda$.
\end{theorem}

\begin{proof}
Consider the superman-kryptonite tournament $T \in \mathcal{T}_n$ where $n \geq 4$.
Observe that $r^{RSEB}_1(T) = 1 - 1/n$ since the superman loses if and only if she is paired with the kryptonite in the first round of the bracket.
Moreover, $r^{RSEB}_n(T) = 0$ since the winner of a bracket must win at least $\ceil{\log_2 n}$ matches.
If the superman and kryptonite collude to make the superman the Condorcet winner, then the superman gains $1/n$ in probability, while the kryptonite's probability of winning remains the same, so $\alpha \geq 1/n$.

To see the result for RKotH, consider the tournament $T \in \mathcal{T}_5$ in which $i \succ_T j$ whenever $i < j$, except both $4, 5 \succ_T 1$.
Note that since RKotH is cover consistent and $4$ covers $5$, we have that $r_5(T) = 0$ .
Meanwhile, $r_1(T) = 2/5$ since 1 wins if and only if the agent who is chosen first is not among 1, $4$, and $5$.
If 1 and $5$ reverse the outcome of their match, then in the resulting tournament $T'$, 1 will win with probability $1/2$ since by DSTC, we can restrict our attention to $T'|_{[4]}$ and 1 wins in $T'|_{[4]}$ if and only if the agent who is chosen first is not among 1 and $4$.
Meanwhile $5$ remains covered in $T'$.
Thus, the superman gains $1/10$ in probability, while the kryptonite's probability of winning remains the same, so $\alpha \geq 1/10$.
Since RKotH is DSTC, this problematic tournament on four agents remains problematic when there are more agents.

Now, to see the result for PR, consider the superman-kryptonite tournament $T$ on four agents.
The associated system of linear equations is
\begin{align*}
    r_1(T) &= \textstyle r_2(T) + \frac{1}{2} r_3(T) \\
    r_2(T) &= \textstyle \frac{1}{2} r_3(T) + \frac{1}{2} r_4(T) \\
    r_3(T) &= \textstyle \frac{1}{2} r_4(T) \\
    r_4(T) &= r_1(T) \\
    r_1(T) + r_2(T) + r_3(T) + r_4(T) &= 1
\end{align*}
The solution to this system is $r_1(T) = 4/13, r_2(T) = 3/13, r_3(T) = 2/13, r_4(T) = 4/13$.
Now, consider the manipulation between teams 1 and 3.
Note that this manipulation simply ``rotates'' the original tournament clockwise.
Thus, letting $T'$ denote the resulting tournament, $r_1(T') = 4/13, r_2(T') = 4/13, r_3(T') = 3/13, r_4(T') = 2/13$.
Note that 3 gains $1/13$ in probability, while 1's probability of winning remains the same, so $\alpha \geq 1/13$.
Since PR is DSTC, this problematic tournament on four agents remains problematic when there are more agents.
\qed
\end{proof}

\begin{theorem}\label{thm:ICR-RDM-RVC-TCR-PRSL}
If ICR satisfies 2-$\nm{\lambda}$-$\alpha$ for some $\lambda \geq 0$, then $\lambda \geq (1 - O(\alpha))\Omega(n^2)$.
If RDM satisfies this property, then $\lambda \geq (1 - O(n \alpha))\Omega(n)$.
If RVC, TCR, or PRSL satisfy this property, then $\lambda \geq (1 - O(\alpha))\Omega(n)$.
\end{theorem}

\begin{proof}
Let $T$ denote the superman-kryptonite tournament on $n$ agents.
Under ICR, the superman wins if and only if the kryptonite is chosen before her in the first $n - 2$ rounds.
Thus, $r^{ICR}_1(T) = \frac{1}{2}\left(1 - \frac{2}{n(n-1)}\right) = \frac{1}{2} - \frac{1}{n(n-1)}$.
Meanwhile, the kryptonite wins if and only if neither her nor the superman are chosen in the first $n - 2$ rounds.
This event happens with probability $r^{ICR}_n(T) = \frac{2}{n(n-1)}$.
Thus, if ICR satisfies 2-$\nm{\lambda}$-$\alpha$, then $\lambda \geq \frac{1 - r^{ICR}_1(T) - \alpha}{r^{ICR}_n(T)} - 1 = \Paren{\frac{1}{4} - \frac{\alpha}{2}}n(n - 1) - 1/2$.

Under RDM, the kryptonite wins if and only if she is not chosen in the first $n - 2$ rounds, so $r^{RDM}_n(T) = \frac{2}{n(n-1)}$.
Note that the superman loses if and only if she is paired with the kryptonite in some round.
The probability that this event occurs if $2/n$, so $r^{RDM}_1(T) = 1 - 2/n$.
Therefore, $\lambda \geq \frac{1 - r^{RDM}_1(T) - \alpha}{r^{RDM}_n(T)} - 1 = \lambda \geq \Paren{1 - \frac{n\alpha}{2}}(n-1) - 1$.

Under RVC, the superman wins if and only if she comes after the kryptonite in the chosen permutation and they are not the first two agents.
Thus, $r^{RVC}_1(T) = \frac{1}{2} - \frac{1}{n(n-1)}$.
The kryptonite on the other hand wins if and only if she comes last in the chosen permutation, so $r^{RVC}_n(T) = 1/n$.
Thus, $\lambda \geq \frac{1 - r^{RVC}_1(T) - \alpha}{r^{RVC}_n(T)} - 1 = \Paren{\frac{1}{2} - \alpha}n - \frac{n-2}{n-1}$.

Under TCR, the superman and kryptonite both win with probability $\frac{1}{n}$ since all agents are in the top cycle.
Then, $\lambda \geq \frac{1 - r^{TCR}_1(T) - \alpha}{r^{TCR}_n(T)} - 1 = (1 - \alpha)n - 2$.

Now, consider the following tournament on $n = 2k + 1$ teams, denoted $0, 1, \dots, n-1$.
Have team $n-1$ defeat team $n-2$.
Have teams $0, 1, \dots, n-3$ lose to team $n-2$ and defeat team $n-1$.
For $i = 0, \dots, n-3$, have team $i$ defeat teams $(i + 1) \bmod{(n-2)}, \dots, (i + k - 1) \bmod{(n-2)}$ and lose to teams $(i-1) \bmod{(n-2)}, \dots, (i-k+1) \bmod{(n-2)}$.
Since teams $0, \dots, n-3$ are indistinguishable from each other (in particular, they will have the same mean return time), $r_0(T) = \dots = r_{n-3}(T)$.
Thus, it suffices to consider the following reduced system:
\begin{align*}
    \textstyle r_{n-2}(T) 
        = \frac{1}{2} r_{n-2}(T) + \frac{2}{n+1} \sum_{i=0}^{n-3} r_i(T) &= \textstyle \frac{1}{2} r_{n-2}(T) + \frac{2(n-2)}{n+1} r_0(T) \\
    r_{n-1}(T) &= \textstyle \frac{1}{2} r_{n-2}(T) + \frac{1}{n-1}r_{n-1}(T) \\
    (n-2)r_0(T) + r_{n-2}(T) + r_{n-1}(T) &= \textstyle \sum_{i=0}^{n-1} r_i(T) = 1
\end{align*}
The solution is
\begin{gather*}
    \textstyle r_0(T) = \dots = r_{n-3}(T) = \frac{n+1}{2(n-1)} \left(\frac{2(n-2)}{n-1} + \frac{(n-2)(n+1)}{2(n-1)} + 1\right)^{-1} \\
    \textstyle r_{n-2}(T) = \frac{2(n-2)}{n-1} \left(\frac{2(n-2)}{n-1} + \frac{(n-2)(n+1)}{2(n-1)} + 1\right)^{-1} \\
    \textstyle r_{n-1}(T) = \left(\frac{2(n-2)}{n-1} + \frac{(n-2)(n+1)}{2(n-1)} + 1\right)^{-1}
\end{gather*}
To conclude, note that if teams $n-2$ and $n-1$ were to manipulate, then team $n-2$ would become the Condorcet winner in the resulting tournament $T'$.
Thus, if PR-SL is 2-$\nm{\lambda}$, then
\[
    \textstyle \lambda \geq \frac{1 - r_{n-2}(T) - \alpha}{r_{n-1}(T)} - 1 \geq \frac{(1 - \alpha)(n-2)(n+1)}{2(n-1)} - 3\alpha
\]
\qed
\end{proof}

\begin{table}
    \caption{Performance summary. 
    Only the strongest fairness properties satisfied by each tournament rule are listed.
    If a stronger or incomparable fairness property is not listed, then the tournament rule does not satisfy it.}
    \centering
    \begin{tabular}{ c || c | c | c | c}
        Rule & Monotone? & Fairness & 2-$\nm{\lambda}$-$\alpha$ & 2-$\nm{\lambda}$-$f(n)$ \\
        \hline
        ICR~\cite{altman2010nonmanipulable} & Yes~\cite{altman2010nonmanipulable} & DSTC & $\lambda \geq \Paren{\frac{1}{2} - \alpha}\Omega(n^2)$ & $f(n) \geq 1/2$ \\
        RVC~\cite{altman2010nonmanipulable} &  Yes~\cite{altman2010nonmanipulable} & DSTC & $\lambda \geq \Paren{\frac{1}{2} - \alpha}\Omega(n)$ & $f(n) \geq 1/2$ \\
        TCR~\cite{altman2010nonmanipulable} & Yes~\cite{altman2010nonmanipulable} & DSTC & $\lambda \geq (1 - \alpha)\Omega(n)$ & $f(n) \geq 1$ \\
        RSEB~\cite{schneider2016condorcetconsistent} & Yes~\cite{schneider2016condorcetconsistent} & TCC & $\alpha \geq 1/n$ & $f(n) \geq \varepsilon(\lambda) > 0$ \\
        RKotH~\cite{schvartzman2019approximately} & Yes~\cite{schvartzman2019approximately} & cover~\cite{schvartzman2019approximately}, DSTC & $\alpha \geq 1/10$ & $f(n) \geq 1/10$ \\
        RDM~\cite{ding2021approximately} & Yes & DSTC & $\lambda \geq \Paren{1 - \frac{n\alpha}{2}}\Omega(n)$ & $f(n) \geq \frac{\lambda - 1}{\lambda(2\lambda + 1)}$ \\
        PR~\cite{pagerank-tournament-rule} & ? & DSTC & $\alpha \geq 1/13$ & $f(n) \geq 1/13$ \\
        PRSL~\cite{pagerank-tournament-rule} & ? & DSTC & $\lambda \geq (1 - \alpha)\Omega(n)$ & $f(n) \geq 1$
    \end{tabular}
    \label{tab:summary}
\end{table}

\section{Reductions}

In a separate direction, we consider fair tournament rules that for fixed $\lambda$ become increasingly non-manipulable with the number of agents.
Formally, we sought rules that satisfy 2-$\nm{\lambda}$-$f(n)$ where $n$ is the number of agents and $f$ is a non-negative, non-increasing function.
Under notions of fairness stronger than Condorcet consistency, it turns out this problem is just as hard as finding rules that satisfy 2-$\nm{\lambda}$-$(\lim_{n \to \infty} f(n))$.

\begin{theorem}\label{thm:DSTC-reduction}
Let $\lambda \geq 0$ and $f \geq 0$ be a non-increasing function such that $f(n) \to \alpha$.
A DSTC tournament rule is 2-$\nm{\lambda}$-$f(n)$ if and only if it is 2-$\nm{\lambda}$-$\alpha$.
\end{theorem}

The idea behind the proof is as follows: by DSTC, the gains from manipulation in a tournament $T$ on $n$ agents are exactly the same as the gains from manipulation among these agents in a larger tournament on $n' > n$ agents in which $T$ is a dominant subtournament.
Thus, the gains from manipulation in $T$ are in fact at most $f(n')$ for all $n' > n$ and hence, at most $\lim_{n' \to \infty} f(n')$.

A similar but weaker result holds for top cycle consistent rules.

\begin{theorem}\label{theorem:TCC_reduction}
Let $\lambda \geq 1$ and $f \geq 0$ be a non-increasing function such that $f(n) \to \alpha$.
There exist a TCC tournament rule satisfying 2-$\nm{\lambda}$-$f(n)$ if and only if there exists a TCC tournament rule satisfying 2-$\nm{\lambda}$-$\alpha$.
\end{theorem}

The proof is similar to that of Theorem \ref{thm:DSTC-reduction}.
However, because we do not have DSTC, we cannot directly relate the gains from manipulation in tournaments on $n$ agents to those in tournaments on $n' > n$ agents.
Nonetheless, we can define a tournament rule on $n$ agents as the limit point of a sequence of tournament rules on $n'$ agents for all $n' > n$.
The gains from manipulation under this limit point will then be at most $\lim_{n' \to \infty} f(n')$.

\begin{theorem}\label{thm:rules-f(n)}
If ICR, RVC, TCR, RKotH, PR, or PRSL satisfy 2-$\nm{\lambda}$-${f(n)}$ for some fixed $\lambda \geq 1$ and some non-increasing function $f \geq 0$, then $f \geq \Omega(1)$.
If RDM satisfies this property, then $f \geq \Omega(1/\lambda)$.
If RSEB satisfies this property, then $f \geq \varepsilon(\lambda)$ where $\varepsilon(\lambda)$ is some strictly positive function of $\lambda$.
\end{theorem}

\section{Discussion}

In this work, we introduced a partially transferable utility model to study the tension between fairness and strategic robustness in the design of tournaments.
In our model, two agents are willing to fix the outcome of their match only if their joint gain is greater than $\lambda$ times any of their losses.
Theorem \ref{thm:lower-bound-all-rules} demonstrates that tournament designers cannot prevent manipulations while maintaining some degree of fairness if agents care about their chances of winning less than twice as much as their opponents' chances.
However, it is possible that caring twice as much is sufficient for the existence of fair and non-manipulable tournament rules.
Unfortunately, we do not know of any tournament rule that achieves this, and Theorems \ref{theorem:RSEB-RKotH-no-lambda} and \ref{thm:ICR-RDM-RVC-TCR-PRSL} show that the tournament rules previously studied in this line of work require agents to care at least $\Omega(n)$ times more about their own chances of winning than their opponents' in order to be non-manipulable.
We leave finding a Condorcet consistent tournament rule that is non-pairwise-manipulable when $\lambda = 1$ as a major open problem.

Theorems \ref{thm:DSTC-reduction} and \ref{theorem:TCC_reduction} may help in resolving this question.
If proving that a DSTC tournament rule witnesses vanishing gains from manipulation is easier than proving that it is non-manipulable (e.g., one can only get upper bounds that approach 0), then Theorem \ref{thm:DSTC-reduction} would imply that the rule is in fact non-manipulable.
On the other hand, if one finds a TCC tournament rule that witnesses vanishing gains from manipulation, then this rule, together with Theorem \ref{theorem:TCC_reduction}, would yield a non-constructive proof that a non-manipulable rule exists.

\begin{credits}
\subsubsection{\ackname} This work was carried out while one of the authors, Eric Xue, was a participant in the 2021 DIMACS REU program at Rutgers University, supported by NSF grant CCF-1852215, under the supervision of Ariel Schvartzman (who at the time was affiliated with DIMACS) and David Pennock. 

\subsubsection{\discintname}
The authors have no competing interests to declare that are
relevant to the content of this article.
\end{credits}
%
%
%
\bibliographystyle{plainnat}
\bibliography{adt2024arxiv}

\appendix
\section{Omitted Proofs}

\begin{proposition}\label{prop:pnm-iff-nm-infty}
A tournament rule is 2-$\pnm$ if and only if it is 2-$\nm{\infty}$.
\end{proposition}

\begin{proof}
Suppose tournament rule $r$ is not 2-$\pnm$.
Then, there exist distinct agents $i, j \in [n]$ and a pair of $\{i,j\}$-adjacent tournaments $T \not= T' \in \mathcal{T}_n$ such that WLOG $r_i(T') - r_i(T) > 0$ and $r_j(T') - r_j(T) \geq 0$.
Thus, 
\[
    r_i(T') + r_j(T') - r_i(T) - r_j(T) > 0 \geq \lim_{\lambda \to \infty}\lambda\max\{r_i(T) - r_i(T'), r_j(T) - r_j(T')\}
\]
so $r$ is not 2-$\nm{\infty}$. 

Conversely, take $r$ to be 2-$\pnm$, so for all distinct agents $i, j \in [n]$ and $\{i,j\}$-adjacent $T \not= T' \in \mathcal{T}_n$, either (1) WLOG $r_i(T') - r_i(T) < 0$, in which case 
\[
    r_i(T') + r_j(T') - r_i(T) - r_j(T) < +\infty = \lim_{\lambda \to \infty} \lambda\max\{r_i(T) - r_i(T'), r_j(T) - r_j(T')\}
\]
or (2) $\max\{r_i(T') - r_i(T), r_j(T') - r_j(T)\} \leq 0$, in which case \[
    r_i(T') + r_j(T') - r_i(T) - r_j(T) \leq 0 \leq \lim_{\lambda \to \infty} \lambda\max\{r_i(T) - r_i(T'), r_j(T) - r_j(T')\}
\]
Thus, $r$ is 2-$\nm{\infty}$.
\end{proof}

\begin{proposition}\label{prop:mono-iff-one-sided-nm}
Let $r$ be a 2-$\nm{\lambda}$ tournament rule for some $\lambda > 0$, then the following two statements are equivalent.
\begin{enumerate}
    \setlength{\itemsep}{0pt}
    \item $r$ is monotone
    \item For all $i \not= j \in [n]$ and $\{i, j\}$-adjacent tournaments $T \not= T' \in \mathcal{T}_n$ such that $i \prec_T j$, $r_i(T') - r_i(T) \leq (\lambda + 1) (r_j(T) - r_j(T'))$
\end{enumerate}
\end{proposition}

\begin{proof}
Let $T \not= T' \in \mathcal{T}_n$ be $\{i, j\}$-adjacent tournaments  such that $i \prec_T j$, and suppose $r$ is monotone.
Then, by monotonicity, $r_i(T) - r_i(T') \leq 0 \leq r_j(T) - r_j(T')$, so
\[
    r_i(T') - r_i(T) \leq r_j(T) - r_j(T') + \lambda \max\{r_i(T) - r_i(T'), r_j(T) - r_j(T')\} = (\lambda + 1) (r_j(T) - r_j(T'))
\]
where the inequality follows from the fact that $r$ is 2-$\nm{\lambda}$.

Now, suppose the second statement holds.
Applying the second statement twice yields
\begin{align*}
    r_i(T') - r_i(T) &\leq (\lambda + 1) (r_j(T) - r_j(T')) \\
    r_j(T) - r_j(T') &\leq (\lambda + 1) (r_i(T') - r_i(T))
\end{align*}
Multiplying the second inequality by $\lambda + 1$, adding the result to the first inequality, and simplifying yields
\[
    \lambda(\lambda + 2)(r_i(T') - r_i(T)) \geq 0
\]
Since $r$ is 2-$\nm{\lambda}$ for $\lambda > 0$, this inequality implies $r_i(T') \geq r_i(T)$.
\end{proof}

\begin{theorem}\label{thm:DSTC-reduction}
Let $\lambda \geq 0$ and $f \geq 0$ be a non-increasing function such that $f(n) \overset{n \to \infty}{\to} \alpha$.
A DSTC tournament rule is 2-$\nm{\lambda}$-$f(n)$ if and only if it is 2-$\nm{\lambda}$-$\alpha$.
\end{theorem}

\begin{proof}
The backward implication is trivial, so we focus on the forward implication.
Suppose a DSTC tournament rule $r$ is not 2-$\nm{\lambda}$-$\alpha$ so that there exist $\varepsilon > 0$, $m \in \NN$, distinct agents $i, j \in [m]$, and a pair of $\{i,j\}$-adjacent tournaments $T \not= T' \in \mathcal{T}_{m}$ such that 
\[
    r_i(T') + r_j(T') - r_i(T) - r_j(T) - \lambda\max\{r_i(T) - r_i(T'), r_j(T) - r_j(T')\} \geq \alpha + \varepsilon
\]
Since $f(n) \overset{n \to \infty}{\to} \alpha$, there exists $n > m$ such that $f(n) < \alpha + \varepsilon$.
Now, consider the tournament $T^{(n)} \in \mathcal{T}_n$ in which $T$ is a dominant sub-tournament.
Let $(T')^{(n)}$ denote the tournament $\{i,j\}$-adjacent to $T^{(n)}$ such that $(T')^{(n)} \not= T^{(n)}$.
Since $r$ satisfies DSTC,
\begin{align*}
    r_i\Paren{(T')^{(n)}} & {} + r_j\Paren{(T')^{(n)}} - r_i\Paren{T^{(n)}} - r_j\Paren{T^{(n)}}  \\
        & {} - \lambda\max\CrBr{r_i\Paren{T^{(n)}} - r_i\Paren{(T')^{(n)}}, r_j\Paren{T^{(n)}} - r_j\Paren{(T')^{(n)}}} \geq \alpha + \varepsilon > f(n)
\end{align*}
so $r$ is not 2-$\nm{\lambda}$-$f(n)$.
\end{proof}

\begin{theorem}\label{theorem:TCC_reduction}
Let $\lambda \geq 1$ and $f \geq 0$ be a non-increasing function such that $f(n) \overset{n \to \infty}{\to} \alpha$.
There exist a TCC tournament rule satisfying 2-$\nm{\lambda}$-$f(n)$ if and only if there exists a TCC tournament rule satisfying 2-$\nm{\lambda}$-$\alpha$.
\end{theorem}

\begin{proof}
Again, the backward implication is trivial, so we focus on the forward implication.
Suppose a tournament rule $r$ satisfies TCC and 2-$\nm{\lambda}$-$f(n)$.
Fix the number of agents $m$.
We will construct a tournament rule on $m$ agents $s^{(m)}$ that is TCC and 2-$\nm{\lambda}$-$\alpha$.

For each $n \in \NN$, define a tournament rule on $m$ agents $w^{m,n}$ as follows.
For all $i \in [m], T \in \mathcal{T}_m$, define
\[
    w^{m,n}_i(T) := r_i\Paren{T^{(n)}}
\]
Note that $TC(T^{(n)}) = TC(T) \subseteq T$, so $w^{m,n}(T) \in \Delta^m$ and since $r$ is TCC and 2-$\nm{\lambda}$-$f(n)$, $w^{m,n}$ is as well.

Now, since $[m] \times \mathcal{T}_m$ is a finite set and each coordinate of $\left((w^{m,n}(T))_{T \in \mathcal{T}_m}\right)_{n=1}^\infty$ is bounded, there exists a convergent subsequence $\left((w^{m,n_k}(T))_{T \in \mathcal{T}_m}\right)_{k=1}^\infty$.
Define $s^{(m)}(T) = \lim_{k\to\infty} w^{n_k}(T)$ for all $T \in \mathcal{T}_m$.
By construction, $s^{(m)}(T) \in \Delta^m$ for all $T \in \mathcal{T}_m$ and $s^{(m)}$ is TCC.
Moreover, $s^{(m)}$ satisfies 2-$\nm{\lambda}$-$\alpha$: for any distinct agents $i, j \in [m]$ and $\{i, j\}$-adjacent tournaments $T, T' \in \mathcal{T}_m$,
\begin{align*}
    s^{(m)}_i(T') + {} & s^{(m)}_j(T') - s^{(m)}_i(T) - s^{(m)}_j(T) - \lambda \max\{s^{(m)}_i(T) - s^{(m)}_i(T'), s^{(m)}_j(T) - s^{(m)}_j(T')\} \\
        = {} & \lim_{k \to \infty} (w^{m,n}_i(T') + w^{m,n}_j(T') - w^{m,n}_i(T) - w^{m,n}_j(T) \\
        & \hspace{25pt} - \lambda \max\{w^{m,n}_i(T) - w^{m,n}_i(T'), w^{m,n}_j(T) - w^{m,n}_j(T')\}) \\
        \leq {} & \lim_{k\to\infty} f(n_k) = \alpha
\end{align*}
Carrying out this procedure for all $m \in \NN$ yields a TCC and 2-$\nm{\lambda}$-$\alpha$ tournament rule $s := \{s^{(m)}\}_{m=1}^\infty$.

\end{proof}

\begin{theorem}\label{theorem:RSEB-RKotH-no-lambda}
Let $\lambda \geq 0$.
If RSEB satisfies 2-$\nm{\lambda}$-$\alpha$, then $\alpha \geq \Omega(1/n)$.
If RKotH or PR satisfy this property, then $\alpha \geq \Omega(1)$.
That is, RSEB, RKotH, and PR are always pairwise manipulable regardless of $\lambda$.
\end{theorem}

\begin{theorem}\label{thm:ICR-RDM-RVC-TCR-PRSL}
If ICR satisfies 2-$\nm{\lambda}$-$\alpha$ for some $\lambda \geq 0$, then $\lambda \geq (1 - O(\alpha))\Omega(n^2)$.
If RDM satisfies this property, then $\lambda \geq (1 - O(n \alpha))\Omega(n)$.
If RVC, TCR, or PRSL satisfy this property, then $\lambda \geq (1 - O(\alpha))\Omega(n)$.
\end{theorem}

\begin{theorem}\label{thm:rules-f(n)}
If ICR, RVC, TCR, RKotH, PR, or PRSL satisfy 2-$\nm{\lambda}$-${f(n)}$ for some fixed $\lambda \geq 1$ and some non-increasing function $f \geq 0$, then $f \geq \Omega(1)$.
If RDM satisfies this property, then $f \geq \Omega(1/\lambda)$.
If RSEB satisfies this property, then $f \geq \varepsilon(\lambda)$ where $\varepsilon(\lambda)$ is some strictly positive function of $\lambda$.
\end{theorem}

\begin{proof}
The results for all rules except RDM and RSEB follow as direct consequences of Theorems \ref{theorem:RSEB-RKotH-no-lambda} and \ref{thm:ICR-RDM-RVC-TCR-PRSL}.
To see the result for RDM, consider the superman kryptonite tournament $T$ on $2\lambda + 1$ players and the distinct $\{1, n\}$-adjacent tournament $T'$ (in which the superman is now the Condorcet winner).
Using the probabilities computed in Theorem \ref{thm:ICR-RDM-RVC-TCR-PRSL}, $f(2\lambda + 1) \geq 1 - r_1^{RDM}(T) - (\lambda + 1)r_j^{RDM}(T) = \frac{2}{2\lambda + 1} - \frac{\lambda + 1}{\lambda(2\lambda + 1)} = \frac{\lambda - 1}{\lambda(2\lambda + 1)}$.
Since RDM is DSTC, this problematic tournament remains problematic when $n \geq 2\lambda + 1$.

The idea behind the proof of the result for RSEB is to assume by way of contradiction that RSEB satisfies 2-$\nm{\lambda}$-${f(n)}$ for some fixed $\lambda \geq 0$ and some $f(n) \overset{n \to \infty}{\to} 0$ and then show that the tournament rule that results from applying the limiting process in Theorem \ref{theorem:TCC_reduction} to RSEB actually requires that $\lambda$ grow with $n$ in order to be 2-$\nm{\lambda}$, which contradicts Theorem \ref{theorem:TCC_reduction}.
We give the details in Section \ref{subsection:RSEB-f(n)} of the appendix.
\end{proof}

\subsection{Gains From Manipulation Under RSEB Do Not Vanish}\label{subsection:RSEB-f(n)}

In what follows, let $r$ denote the RSEB tournament rule.
At times, it is useful to consider a tournament of size $k > n$ in which $T \in \mathcal{T}_n$ is a dominant sub-tournament, so let $T^{(k)}$ denote some tournament in $\mathcal{T}_k$ in which $i \succ_T j \iff i \succ_{T^{(k)}} j$ for all $i, j \in T$ and $i \succ_{T^{(k)}} j$ for all $i \in T, j \in T^{(k)} \setminus T$.
We prove the following result.

\begin{theorem}
If RSEB satisfies 2-$\nm{\lambda}$-${f(n)}$ for some fixed $\lambda \geq 1$ and some non-increasing function $f \geq 0$, then $f \geq \varepsilon(\lambda)$ where $\varepsilon(\lambda)$ is some strictly positive function of $\lambda$.
\end{theorem}

Suppose by way of contradiction that RSEB is 2-$\nm{\lambda}$-${f(n)}$ for some fixed $\lambda \geq 0$ and some non-increasing function $f \geq 0$ such that $f(n) \to 0$ as $n \to \infty$.
Fix $n = 2^h$.
Since RSEB is TCC and assumed to be 2-$\nm{\lambda}$-${f(n)}$, following the logic in Theorem \ref{theorem:TCC_reduction}, we can find a convergent subsequence $((r_1(T^{(2^{m_j})}), \dots, r_n(T^{(2^{m_j})})))_{j=1}^\infty$ such that the limit point $s^{(n)}$ is 2-$\nm{\lambda}$.

Now, let $T \in \mathcal{T}_n$ be the superman kryptonite tournament on $n = 2^h$ agents and let $T'$ denote the distinct tournament that is $\{1, n\}$-adjacent to $T$.
Since $s^{(n)}$ is 2-$\nm{\lambda}$,
\begin{align*}
    0 \geq 1 - s^{(n)}_1(T) - \lambda s^{(n)}_n(T) 
        &\geq \lim_{j\to\infty} (1 - r_1(T_{2^{m_j}}) - \lambda r_n(T_{2^{m_j}})) \\
        &\geq \frac{1}{3n} - \frac{\lambda}{2^{n-1}-1} \tag{Lemmas \ref{lemma:upper-bound-kryptonite-probability} and \ref{lemma:limiting-probability-superman-loses}}
\end{align*}
Thus,
\[
    \lambda \geq \frac{2^{n-1}-1}{3n}
\]
Since our choice of convergent subsequence was arbitrary, this inequality holds for all limit points $s^{(n)}$.
Carrying out this analysis for all $n = 2^h$, we get that any tournament rule $s$ that arises from carrying out the limiting procedure from the proof of Theorem \ref{theorem:TCC_reduction} on RSEB requires that $\lambda$ grows exponentially in the number of agents in order to be 2-$\nm{\lambda}$, contradicting Theorem \ref{theorem:TCC_reduction}.
Thus, RSEB is not 2$\nm{\lambda}$-${f(n)}$ for a fixed $\lambda \geq 0$ and $f(n) \to 0$ as $n \to \infty$.

\begin{lemma}\label{lemma:rseb-recurrence}
Let $n = 2^h$ and let $T \in \mathcal{T}_n$ be the superman-kryptonite tournament on $n$ agents.
For all $m \geq h + 1$,
\[
    r_n(T_{2^m}) = \frac{2}{\binom{2^m}{2^{m-1}}} \left[\binom{2^m-n}{2^{m-1}-1} + \binom{2^m-n}{2^{m-1}-n} r_n(T_{2^{m-1}}) \right]
\]
where $T^{(2^m)} \in \mathcal{T}_{2^m}$ is a tournament on $2^m$ agents in which $T$ is a dominant subtournament. 
\end{lemma}

\begin{proof}
For any tournament $S$, consider the following approach for computing $r_i(S)$.
For a given partition of $[n]$ into two equal sets $A$ and $B$ in which $i \in A$, the probability that $i$ wins a bracket in which the players on one side are in $A$ and the players on the other are in $B$ is
\[
    r_i(S|_A) \sum_{k: i \succ_S k} r_k(T|_B) 
\]
where $S|_A$ is the tournament subgraph of $S$ induced by the players in $A$.
In other words, the probability that $i$ wins such a bracket is the probability that she wins her side of the bracket times the probability that someone she can beat the winner of the other side.
Randomizing over all such partitions,
\[
    r_i(S) = \frac{2}{\binom{n}{n/2}} \sum_{A,B} \left(r_i(S|_A) \sum_{k: i \succ_S k} r_k(S|_B) \right)
\]

Now, to prove the recurrence relation, we consider several cases of partitions.
If $n \in A$, while $[n-1] \subseteq B$, then $r_n(T_{2^m}|_A) = r_1(T_{2^m}|_B) = 1$ since $n$ can defeat all the dummy players on her side of the bracket, and $2$ only loses to $n$, who is on the other side of the bracket.
If $[n] \subseteq A$, then $r_n(T_{2^m}|_A) = r_n(T_{2^{m-1}})$ and $\sum_{k: n \succ_T k} r_k(T|_B) = 1$ since $n$ can defeat all the dummy players on the other side.
If $n \in A, 1 \in B$, and there exists $i \in A \cap [n-1]$, then $r_n(T|_A) = 0$ since $n$ will have to face someone she loses to before the finals.
Otherwise, $1, n \in A$ and there exists $i \in B \cap [n-1]$.
In this case, $\sum_{k: n \succ_T k} r_k(T|_B) = 0$ since the winner of the $B$ side of the bracket will be someone $n$ loses to.
Thus,
\begin{align*}
    r_n(T_{2^m}) = \frac{2}{\binom{2^m}{2^{m-1}}} \left(\sum_{n \in A, [n-1] \subseteq B} 1 + \sum_{[n] \in A} r_n(T_{2^{m-1}}) \right)
\end{align*}
The number of partitions up to symmetry such that $n \in A, [n-1] \subseteq B$ is $\binom{2^m-n}{2^{m-1}-1}$ (fix $n \in A$ and choose the remaining $2^{m-1}-1$ players from outside of $[n]$), and the number of partitions such that $[n] \subseteq A$ is $\binom{2^m-n}{2^{m-1}}$ (fix $[n] \subseteq A$ and choose the remaining $2^{m-1}-n$ players from outside of $[n]$).
The recurrence relation now follows.
\end{proof}

\begin{lemma}\label{lemma:upper-bound-kryptonite-probability}
Let $n = 2^h$ and let $T \in \mathcal{T}_n$ be the superman-kryptonite tournament on $n$ agents.
For all $m \geq h + 1$,
\[
    r_n\Paren{T^{(2^m)}} \leq \frac{1}{2^{n-1}-1}
\]
where $T^{(2^m)} \in \mathcal{T}_{2^m}$ is some tournament in which $T$ is a dominant subtournament. 
\end{lemma}

\begin{proof}
We show by induction that $r_n(T_{2^m}) \leq \sum_{k=1}^{m-h} \frac{1}{2^{k(n-1)}}$.
For our base case, consider $m = h+1$, i.e., $2^m = 2n$.
Note that the kryptonite $n$ cannot win in the superman kryptonite tournament of size $n$, so $r_n(T_n) = 0$.
Thus, 
\begin{align*}
    r_n(T_{2n}) &= \frac{2}{\binom{2n}{n}} \binom{n}{n-1} \\
        &= \frac{n!n!}{(2n-1)!} \\
        &= \frac{(n-1)!}{\prod_{k=1}^{n-1}(2n-k)} \\
        &= \frac{(n-1)!}{2^{n-1}\prod_{k=1}^{n-1}(n-k/2)} \\
        &= \frac{1}{2^{n-1}} \prod_{k=1}^{n-1} \frac{n-k}{n-k/2} \leq \frac{1}{2^{n-1}} 
\end{align*}
as desired.

Now, suppose that $r_n(T_{2^{m-1}}) \leq \sum_{k=1}^{m-h-1} \frac{1}{2^{k(n-1)}}$, and consider $r_n(T_{2^m})$.
For notational convenience, let $N := 2^{m-1}$.
Observe that
\begin{align*}
    \frac{2}{\binom{2N}{N}} \binom{2N-n}{N-1} 
        &= \frac{2N!N!}{(2N)!} \cdot \frac{(2N-n)!}{(N-1)!(N-n+1)!} \\
        &= \frac{\prod_{k=0}^{n-2} (N-k)}{\prod_{k=1}^{n-1} (2N-k)} \\
        &= \frac{1}{2^{n-1}}\prod_{k=1}^{n-1} \frac{N-k+1}{N-k/2} \leq \frac{1}{2^{n-1}}
\end{align*}
The last line comes from the fact that $\frac{N}{N-1/2} \cdot \frac{N-3}{N-2} \leq 1$, and the rest of the terms in the product are $\leq 1$.
Moreover, 
\begin{align*}
    \frac{2}{\binom{2N}{N}} \binom{2N-n}{N-n} 
        &= \frac{2N!N!}{(2N)!} \cdot \frac{(2N-n)!}{(N-n)!N!} \\
        &= \frac{\prod_{k=1}^{n-1}(N-k)}{\prod_{k=1}^{n-1}(2N-k)} \\
        &= \frac{1}{2^{n-1}} \prod_{k=1}^{n-1} \frac{N-k}{N-k/2} \leq \frac{1}{2^{n-1}}
\end{align*}
Thus,
\begin{align*}
    r_n(T_{2^m}) &= \frac{2}{\binom{2^m}{2^{m-1}}} \left[\binom{2^m-n}{2^{m-1}-1} + \binom{2^m-n}{2^{m-1}-n} r_n(T_{2^{m-1}}) \right] \\
    &\leq \frac{1}{2^{n-1}} + \frac{r_n(T_{2^{m-1}})}{2^{n-1}} \\
    &= \frac{1}{2^{n-1}} + \sum_{k=1}^{m-h-1} \frac{1}{2^{(k+1)(n-1)}} \tag{inductive hypothesis} \\
    &= \sum_{k=1}^{m-h} \frac{1}{2^{k(n-1)}}
\end{align*}
as desired.
We conclude by remarking that $\sum_{k=1}^\infty \frac{1}{2^{k(n-1)}} = \frac{1}{2^{n-1}-1}$.
\end{proof}

\begin{lemma}\label{lemma:series_lb}
Let $n = 2^h$.
There exists $\ell$ such that $\sum_{k=1}^\ell (1 - \frac{1}{2^k})^{n-2}\frac{1}{2^k} \geq \frac{1}{3n}$.
\end{lemma}

\begin{proof}
Consider the series $\sum_{k=1}^\infty (1 - \frac{1}{2^k})^{n-1}\frac{1}{2^k}$.
The partial sums in this series are increasing and upper bounded by $\sum_{k=1}^\infty \frac{1}{2^k} = 1$, so the series converges.
Define $f(x) := (1 - \frac{1}{2^x})^{n-1}\frac{1}{2^x}$.
\[
    f'(x) = \frac{\ln(2)(1-1/2^x)^n(n-2^x)}{(2^x-1)^2}
\]
Thus, $f$ is strictly increasing on $[0,h]$ and strictly decreasing on $[h,\infty]$.
Moreover, $f$ achieves a local maximum at $h$.
Therefore,
\begin{align*}
    \Paren{1 - \frac{1}{2^h}}^{n-1}\frac{1}{2^h} + \sum_{k=1}^\infty \Paren{1 - \frac{1}{2^k}}^{n-1}\frac{1}{2^k} 
        &\geq \int_{0}^\infty f(x) \; \mathrm{d}x \\
        &= \left[\frac{(1-1/2^x)^n}{n \ln(2)}\right]_0^\infty \\
        &= \frac{1}{n \ln(2)}
\end{align*}
so
\[
    \sum_{k=1}^\infty \Paren{1 - \frac{1}{2^k}}^{n-1}\frac{1}{2^k}  \geq \frac{1}{n \ln(2)} - \Paren{1 - \frac{1}{n}}^{n-1}\frac{1}{n} \geq \frac{1}{n \ln(2)} - \frac{1}{n}
\]
Since $\frac{1 - \ln(2)}{n\ln(2)} > \frac{1}{3n}$, there exists $\ell$ such that $\sum\limits_{k=1}^\ell (1 - \frac{1}{2^k})^{n-2}\frac{1}{2^k} > \sum\limits_{k=1}^\ell (1 - \frac{1}{2^k})^{n-1}\frac{1}{2^k} \geq \frac{1}{3n}$.
\end{proof}

\begin{lemma}\label{lemma:limiting-probability-superman-loses}
Let $n = 2^h$. 
Let $T \in \mathcal{T}_n$ be the superman-kryptonite tournament on $n$ agents and let $T^{(2^m)} \in \mathcal{T}_{2^m}$ be some tournament in which $T$ is a dominant subtournament.
If $(r_1(T_{2^{m_i}}))_{i=1}^\infty$ is a convergent (sub)sequence, then
\[
    \lim_{i\to\infty} 1 - r_1(T_{2^{m_i}}) \geq \frac{1}{3n}
\]
\end{lemma}

\begin{proof}
Note that $1 - r_1(T_{2^m})$ is the probability that the superman (in $T$) loses $T_{2^m}$.
Since the superman loses if and only if she encounters the kryptonite in some round, for $m \geq h + 1$,
\[
    1 - r_1(T_{2^m}) = \sum_{k=0}^{m-1} \frac{\binom{2^m-2^h}{2^k-1}}{\binom{2^m-1}{2^k}}
\]
To see why this expression is correct, observe that
\[
    \frac{\binom{2^m-2^h}{2^k-1}}{\binom{2^m-1}{2^k}}
\]
is the probability that the superman encounters the kryptonite in round $k + 1$.
Given the seed position of the superman, there is exactly one subtree of height $k + 1$ that the kryptonite must be seeded into in order to have a chance of facing the superman in round $k + 1$.
Moreover, the kryptonite faces the superman in the desired round if and only if the remaining $2^k-1$ players in this subtree are dummy players.
Thus, the probability that the kryptonite reaches round $k$ is the probability that the $2^k$ players in this subtree consist of the kryptonite and $2^k-1$ dummy players.
Now, re-index to get
\[
    1 - r_1(T_{2^m}) = \sum_{k=1}^{m} \frac{\binom{2^m-2^h}{2^{m-k}-1}}{\binom{2^m-1}{2^{m-k}}}
\]

Now, consider a term in the sum.
\begin{align*}
    \frac{\binom{2^m-2^h}{2^{m-k}-1}}{\binom{2^m-1}{2^{m-k}}}
        &= \frac{(2^m - 2^h)!}{(\frac{2^m}{2^{k}}-1)!(\frac{2^m(2^{k}-1)}{2^{k}} - 2^h + 1)!} \cdot \frac{(\frac{2^m}{2^{k}})!(\frac{2^m(2^{k}-1)}{2^{k}}-1)!}{(2^m-1)!} \\
        &= \prod_{j=1}^{2^h-2} \frac{\frac{2^m(2^{k}-1)}{2^{k}}-j}{2^m-j} \cdot \frac{\frac{2^m}{2^{k}}}{2^m - 2^h + 1} \\
        &= \left(\frac{2^{k}-1}{2^{k}}\right)^{2^h-2} \frac{1}{2^{k}} \prod_{j=1}^{2^h-2} \frac{2^m - \frac{2^{k}j}{(2^{k}-1)}}{2^m-j} \cdot \frac{2^m}{2^m - 2^h + 1} 
\end{align*}
Thus,
\[
    \lim_{m\to\infty} \frac{\binom{2^m-2^h}{2^{m-k}-1}}{\binom{2^m-1}{2^{m-k}}} = \left(\frac{2^{k}-1}{2^{k}}\right)^{n-2} \frac{1}{2^{k}} = \left(1 - \frac{1}{2^{k}}\right)^{n-2} \frac{1}{2^{k}} \tag{$n = 2^h$}
\]

Now, let $\varepsilon > 0$ be given.
By Lemma \ref{lemma:series_lb}, there exists $\ell$ such that 
\[
    \sum_{k=1}^\ell (1 - \frac{1}{2^k})^{n-2}\frac{1}{2^k} \geq \frac{1}{3n}
\]
Fix $\ell$.
Since for each $1 \leq k \leq \ell$,
\[
    \lim_{m\to\infty} \frac{\binom{2^m-2^h}{2^{m-k}-1}}{\binom{2^m-1}{2^{m-k}}} = \left(1 - \frac{1}{2^{k}}\right)^{n-2} \frac{1}{2^{k}}
\]
there exists $M_k \geq \ell$ such that for all $m \geq M_k$, 
\[
    \abs{\frac{\binom{2^m-2^h}{2^{m-k}-1}}{\binom{2^m-1}{2^{m-k}}} - (1 - \frac{1}{2^k})^{n-2}\frac{1}{2^k}} < \frac{\varepsilon}{\ell}
\]
Let $M := \max_{1 \leq k \leq \ell} M_k$, so that for all $m \geq M$,
\begin{align*}
    \abs{\sum_{k=1}^\ell \frac{\binom{2^m-2^h}{2^{m-k}-1}}{\binom{2^m-1}{2^{m-k}}} - \sum_{k=1}^\ell (1 - \frac{1}{2^k})^{n-2}\frac{1}{2^k}} 
        &\leq \sum_{k=1}^\ell \abs{\frac{\binom{2^m-2^h}{2^{m-k}-1}}{\binom{2^m-1}{2^{m-k}}} - (1 - \frac{1}{2^k})^{n-2}\frac{1}{2^k}} \\
        &< \sum_{k=1}^\ell \frac{\varepsilon}{\ell} = \varepsilon
\end{align*}
Thus, 
\[
    \lim_{i\to\infty} 1 - r_1(T_{2^{m_i}}) = \lim_{i\to\infty} \sum_{k=1}^{m_i} \frac{\binom{2^{m_i}-2^h}{2^{{m_i}-k}-1}}{\binom{2^{m_i}-1}{2^{{m_i}-k}}} \geq \lim_{m\to\infty} \sum_{k=1}^\ell \frac{\binom{2^m-2^h}{2^{m-k}-1}}{\binom{2^m-1}{2^{m-k}}} = \sum_{k=1}^\ell (1 - \frac{1}{2^k})^{n-2}\frac{1}{2^k} \geq \frac{1}{3n}
\]
\end{proof}




\end{document}